\documentclass[a4paper,10pt]{article}
\usepackage{fullpage}
\usepackage{comment}
\usepackage[utf8]{inputenc}
\usepackage{enumerate}
\usepackage[OT4]{fontenc}
\usepackage{graphicx}
\usepackage{amssymb}
\usepackage{amsmath}
\usepackage{amsfonts}
\usepackage{amsthm}
\usepackage{algorithm}
\usepackage[noend]{algpseudocode}
\usepackage{algorithmicx}
\usepackage{tikz}
\usepackage[textsize=scriptsize]{todonotes}
\usepackage[hidelinks]{hyperref}
\usepackage[capitalise]{cleveref}
\usepackage{caption}
\usepackage{tikzscale}
\usepackage{authblk}
\usepackage{enumitem}
\usepackage{todonotes}
\usepackage{thmtools}
\usepackage{thm-restate}
\usetikzlibrary{backgrounds}
\usetikzlibrary{patterns}
\usetikzlibrary{decorations.pathmorphing}
\usetikzlibrary{decorations.pathreplacing}

\let\originalleft\left
\let\originalright\right
\renewcommand{\left}{\mathopen{}\mathclose\bgroup\originalleft}
\renewcommand{\right}{\aftergroup\egroup\originalright}

\newtheorem{theorem}{Theorem}[section]
\newtheorem{lemma}[theorem]{Lemma}

\newtheorem{fact}[theorem]{Fact}
\newtheorem{definition}[theorem]{Definition}

\newcommand{\Ot}{\ensuremath{\widetilde{O}}}

\newcommand{\eps}{\epsilon}
\newcommand{\dist}{\delta}

\newcommand{\rev}[1]{\ensuremath{#1}^{\mathrm{R}}}

\newcommand{\nil}{\ensuremath{\mathbf{nil}}}

\newcommand{\round}{\ensuremath{\mathrm{round}}}

\newcommand{\exc}{\mathrm{exc}}
\newcommand{\totexc}{\Psi}

\begin{document}

\author[1]{Adam Karczmarz\thanks{Supported by ERC Consolidator
Grant 772346 TUgbOAT and the Polish National Science Centre 2018/29/N/ST6/00757 grant.}}
\author[2]{Piotr Sankowski\thanks{Supported by ERC Consolidator Grant 772346 TUgbOAT.}}

\affil[1]{Institute of Informatics, University of Warsaw, Poland}
\affil[ ]{\texttt{a.karczmarz@mimuw.edu.pl} \medskip}

\affil[2]{Institute of Informatics, University of Warsaw, Poland}
\affil[ ]{\texttt{sank@mimuw.edu.pl}}

\title{Min-Cost Flow in Unit-Capacity Planar Graphs}
\date{}

\maketitle
  
  \begin{abstract}
  In this paper we give an $\Ot((nm)^{2/3}\log C)$ time algorithm for computing min-cost flow (or min-cost
circulation) in unit capacity planar multigraphs where edge costs are integers bounded by $C$.
  For planar multigraphs, this improves upon the best known
  algorithms for general graphs: the $\Ot(m^{10/7}\log C)$ time
algorithm of Cohen et al. [SODA 2017], the $O(m^{3/2}\log(nC))$ time algorithm of Gabow and Tarjan [SIAM J. Comput. 1989] and the $\Ot(\sqrt{n}m \log C)$ time
algorithm of Lee and Sidford [FOCS 2014].
In particular, our result constitutes the first known fully combinatorial algorithm that breaks the $\Omega(m^{3/2})$ time barrier for min-cost
flow problem in planar graphs.

To obtain our result we first give a very simple successive shortest paths based scaling
  algorithm for unit-capacity min-cost flow problem that does not explicitly
  operate on dual variables. This algorithm also runs in $\Ot(m^{3/2}\log{C})$
  time for general graphs, and, to the best of our knowledge, it has not been described before.
  We subsequently show how to implement this algorithm faster on planar graphs
  using well-established tools: $r$-divisions and efficient algorithms for computing (shortest) paths in so-called dense
  distance graphs.

\end{abstract}

\section{Introduction}
The min-cost flow is the core combinatorial optimization problem that now has been studied for over 60 years, starting with the work of Ford and Fulkerson~\cite{ff58}.
Classical combinatorial algorithms for this problem have been developed in the 80s.
Goldberg and Tarjan~\cite{GoldbergT90} showed an $\Ot(nm\log{C})$ time weakly-polynomial algorithm
for the case when edge costs are integral, and where $C$ is the maximum edge cost.
Orlin~\cite{Orlin88} showed the best-known strongly polynomial time algorithm running in $\Ot(m^2)$ time.
Faster weakly-polynomial algorithms have been developed in this century
using interior-point methods: Daitch and Spielman~\cite{DaitchS08} gave an $\Ot(m^{3/2}\log{(U+C)})$
algorithm, and later Lee and Sidford~\cite{LeeS14} obtained an $\Ot(\sqrt{n}m\log{(U+C)})$
algorithm, where $U$ is the maximum (integral) edge capacity.

Much attention has been devoted to the unit-capacity case
of the min-cost flow problem.
Gabow and Tarjan~\cite{GabowT89} gave a $O(m^{3/2}\log{(nC)})$ time algorithm.
Lee and Sidford~\cite{LeeS14} matched this bound up to polylogarithmic factors
for $m=\Ot(n)$, and improved upon it for larger densities, even though
their algorithm solves the case of arbitrary integral capacities.
Gabow and Tarjan's result remained the best known bound
for more than 28 years -- the problem
witnessed an important progress only very recently.
In 2017 an algorithm that breaks the $\Omega(m^{3/2})$ time barrier for min-cost flow problem was given by Cohen et al.~\cite{CohenMSV17}.
This algorithm runs in $\Ot(m^{10/7}\log C)$ time
and is also based on interior-point methods.

It is worth noting that currently the algorithms of \cite{CohenMSV17, LeeS14} constitute
the most efficient solutions for the entire range of possible
densities (up to polylogarithmic factors) and
are also the best-known algorithms
for important special cases, e.g., planar graphs or minor-free graphs.
Both of these solutions
are based on interior point methods and do not shed light on the combinatorial structure of the problem.

In this paper we study the unit-capacity min-cost flow in planar multigraphs.
We improve upon \cite{CohenMSV17, LeeS14} by giving the first known $\Ot((mn)^{2/3}\log C)=\Ot(m^{4/3}\log{C})$ time algorithm for computing min-cost $s,t$-flow
and min-cost circulation in planar multigraphs.\footnote{It is known that simple planar graphs have $O(n)$ edges. However,
multiple parallel edges (with possibly different costs) are useful in the unit-capacity
min-cost flow problem, as they allow us to encode larger edge capacities. Therefore, in this paper we work with planar multigraphs.} 
Our algorithm is fully combinatorial and uses the scaling approach of Goldberg and Tarjan~\cite{GoldbergT90}.
At each scale it implements the classical shortest augmenting path approach
similar to the one known from the well-known Hopcroft-Karp algorithm for maximum bipartite matching~\cite{HopcroftK73}.

\paragraph{Related work.}
Due to immense number of works on flows and min-cost flows we will not review all of them. Instead we concentrate only on the ones that are relevant to the sparse and planar
graph case, as that is the regime where our results are of importance. As already noted above the fastest algorithms for min-cost flows in planar multigraphs are implied
by the algorithms for general case. This, however, is not the case for maximum flow problem. Here, the fastest algorithms are based on planar graph duality and reduce
the problem to shortest path computations. The undirected $s,t$-flow problem can be solved in $O(n \log \log n)$ time~\cite{ItalianoNSW11}, whereas the directed $s,t$-flow problem can be
solved in $O(n \log n)$ time~\cite{BorradaileK09, Erickson10}. Even for the case with multiple source and sinks, a nearly-linear time algorithm is known~\cite{BorradaileKMNW17}.

These results naturally raise as an open question whether similar nearly-linear bounds could be possible for min-cost flow. Until very recently there
has been no progress towards answering this open question.
Partial progress was made by devising $\Ot(n^{4/3}\log{C})$ time~\cite{AsathullaKLR18} and $\Ot(n^{6/5}\log{C})$ time~\cite{LahnR19} algorithms
for min-cost perfect matchings in bipartite planar graphs.
Lahn and Raghvendra also give an $\Ot(n^{7/5}\log{C})$ time minimum cost perfect matching algorithm
for minor-free graphs. These algorithms can be seen as specialized versions of the Gabow-Tarjan's algorithm
for the assignment problem~\cite{GabowT89}.

Gabow and Tarjan~\cite{GabowT89} reduced min-cost flow problem to so-called min-cost
perfect degree-constrained subgraph problem on a bipartite multigraph,
which they solved by extending their algorithm for minimum cost perfect matching.
Hence it seems plausible that the recent algorithm of Lahn and Raghvendra~\cite{LahnR19}
can be extended to solve min-cost flow, since their algorithm builds upon the Gabow-Tarjan algorithm.
The reduction presented by Gabow and Tarjan \emph{is not} planarity preserving, though.
Nevertheless, min-cost perfect matching problem can
be reduced to min-cost flow problem in an efficient and planarity preserving way~\cite{MillerN95}.
The opposite reduction can be done in planarity preserving way as recently shown~\cite{Sankowski18}.
However, this reduction is not efficient and produces a graph of quadratic size.
Hence, we cannot really take advantage of it.
\paragraph{Overview and comparison to \cite{AsathullaKLR18, LahnR19}.} We concentrate on the \emph{min-cost circulation}
problem, which is basically the min-cost flow problem with all vertex demands equal to $0$.
It is well-known \cite{GoldbergHKT17} that the min-cost $s,t$-flow problem can be solved
by first computing some $s,t$-flow
$f$ of requested value (e.g., the maximum value), and then finding a min-cost
circulation on the residual network $G_f$.
This reduction is clearly planarity-preserving.
Since an $s,t$-flow of any given value (in particular, the maximum value) can be found in a planar graph
in nearly-linear time (see \cite{Erickson10}), this reduction works in nearly-linear time as well.

Our min-cost circulation algorithm resembles
the recent works on minimum cost planar perfect matching~\cite{AsathullaKLR18, LahnR19},
in the sense that we simulate some already-good scaling algorithm for general
graphs, but implement it more efficiently using the known and well-established tools from the area
of planar graph algorithms.
However, instead of simulating an existing unit-capacity min-cost
flow algorithm, e.g., \cite{GabowT89, GoldbergHKT17}, we
use a very simple successive-shortest paths based algorithm that,
to the best our knowledge, has not been described before.

Our algorithm builds upon the cost-scaling framework
of Goldberg and Tarjan~\cite{GoldbergT90}, similarly as the recent simple unit-capacity min-cost flow algorithms of Goldberg et al. \cite{GoldbergHKT17}.
In this framework, a notion of $\eps$-optimality of a flow
is used.
A flow $f$ is $\eps$-optimal wrt. to a price
function $p$ if for any edge $uv=e\in E(G_f)$ we have $c(e)-p(u)+p(v)\geq -\eps$.

Roughly speaking, the parameter $\eps$ measures the quality of a circulation:
any circulation is trivially $C$-optimal wrt. $p$, whereas any $\frac{1}{n}$-feasible (wrt. $p$) circulation is
guaranteed to be optimal.
The general scheme is to start with a $C$-optimal circulation, run $O(\log(nC))$
\emph{scales} that improve the quality of a circulation by a factor of $2$, and this way obtain the optimal solution.

  \newcommand{\source}{\ensuremath{s}}
  \newcommand{\sink}{\ensuremath{t}}

We show that a single scale can be solved by repeatedly sending flow along
a cheapest $\source\to\sink$ path in a certain graph $G_f''$ with a single source
$\source$ and a single sink $\sink$, that approximates the residual graph $G_f$.
Moreover, if we send flow
simultaneously along a maximal set of cheapest $\source\to\sink$ paths at once,
like in~\cite{EvenT75, HopcroftK73},
we finish after $O(\sqrt{m})$ augmentations.
However, as opposed to $\cite{EvenT75, HopcroftK73}$, our graph $G_f''$ is weighted and might have negative
edges.
We overcome this difficulty as in the classical successive shortest path
approach for min-cost flow, by using distances from the previous flow augmentation
as a feasible price function that can speed-up next shortest path computation.
Our algorithm also retains a nice property\footnote{Gabow-Tarjan algorithm for min-cost bipartite matching has a similar property, which
was instrumental for obtaining the recent results on minimum-cost planar bipartite matching~\cite{AsathullaKLR18, LahnR19}} of the Even-Tarjan algorithm that
the total length (in terms of the number of edges) of all the used augmenting
paths is $O(m\log{m})$.

The crucial difference between our per-scale procedure
and those of~\cite{GabowT89, GoldbergHKT17} is that
we do not ``adjust'' dual variables $p(v)$ at all while the procedure runs:
we only use them to compute $G_f''$, and recompute them from scratch in nearly-linear
time when the procedure finishes.
In particular, the recent results of \cite{AsathullaKLR18, LahnR19} are quite complicated
since, in order to simulate the Gabow-Tarjan algorithm~\cite{GabowT89}, they impose
and maintain additional invariants about the duals.

The only bottlenecks of our per-scale procedure are (1) shortest paths computation, (2)
picking a maximal set of edge-disjoint $s\to t$ paths in an unweighted graph\footnote{This is sometimes called \emph{the blocking flow} problem
and can be solved for unit capacities in linear time.}.

We implement these on a planar network using standard methods. Let $r\in [1,n]$ be some parameter.
We construct a \emph{dense distance graph} $H_f''$ (e.g., \cite{FR06, GawrychowskiK18})
built upon an $r$-division (e.g., \cite{KleinMS13}) of $G_f''$.
The graph $H_f''$ is a compressed representation of the distances in $G_f''$
with $O(n/\sqrt{r})$ vertices and $O(m)$ edges. 
Moreover, it can be updated in $\Ot(r)$ time per edge used by the flow.
Hence, the total time spent on updating $H_f''$ is $\Ot(mr)$.
As we show, running our per-scale procedure on $H_f''$ is sufficient to
simulate it on $G_f''$. Computing distances in a dense distance
graph requires $\Ot(n/\sqrt{r})$ time~\cite{FR06, GawrychowskiK18}.
To complete the construction, we show how to find a maximal set of edge-disjoint paths in $\Ot(n/\sqrt{r})$ amortized time.
To this end, we also exploit the properties
of reachability in a dense distance graph, used previously
in dynamic reachability algorithms for planar digraphs \cite{ItalianoKLS17, Karczmarz18}.
This way, we obtain $\Ot(\sqrt{m}n/\sqrt{r}+mr)$ running time per scale.
This is minimized roughly when $r=n^{2/3}/m^{1/3}$.

Recall that Lahn and Raghvendra~\cite{LahnR19} obtained
a polynomially better (than ours) bound of \linebreak
$\Ot(n^{6/5}\log{C})$,
but only for planar min-cost perfect matching problem.
To achieve that, they use an additional idea due to Asathulla et al. \cite{AsathullaKLR18}. Namely,
they observe that by introducing vertex weights, one
can make augmenting paths avoid edges incident to boundary vertices, thus making
the total number of pieces ``affected'' by augmenting paths
truly-sublinear in $n$.
It~is not clear how to apply this idea to the min-cost flow
problem without making additional assumptions about the structure
of the instance, like bounded-degree (then, there are only $O(n/\sqrt{r})$
edges incident to boundary vertices of an $r$-division), or
bounded vertex capacities (so that only $O(1)$ units of flow
can go through each vertex; this is satisfied in the perfect matching case).
This phenomenon seems not very surprising once we recall that
such assumptions lead to better bounds even for general graphs:
the best known combinatorial algorithms for min-cost
perfect matching run in $O(n^{1/2}m\log{(nC)})$ time, whereas
for min-cost flow in $O(m^{3/2}\log{(nC)})$ time~\cite{GabowT89, GoldbergHKT17}.

\paragraph{Organization of the paper.}
In Section~\ref{s:prelims} we introduce the notation and describe the scaling
framework of~\cite{GoldbergT90}.
Next, in Section~\ref{s:general}, we describe the per-scale procedure
of unit-capacity min-cost flow for general graphs.
Finally, in Section~\ref{s:planar} we give our algorithm for planar graphs.

%
\section{Preliminaries}\label{s:prelims}
  Let $G_0=(V,E_0)$ be the input directed multigraph.
  Let $n=|V|$ and $m=|E_0|$.
  Define $G=(V,E)$ to be a multigraph such that $E=E_0\cup\rev{E_0}$, $E_0\cap \rev{E_0}=\emptyset$, where $\rev{E_0}$ is the set of \emph{reverse edges}.
  For any $uv=e\in E$, there is an edge $\rev{e}\in E$ such that $\rev{e}=vu$ and $\rev{(\rev{e})}=e$.
  We have $e\in E_0$ iff $\rev{e}\in \rev{E_0}$.

  Let $u:E_0\to \mathbb{R}_+$ be a \emph{capacity function}.
  A \emph{flow} is a function $f:E\to \mathbb{R}$ such that for any $e\in E$
  $f(e)=-f(\rev{e})$ and for each $e\in E_0$, $0\leq f(e)\leq u(e)$.
  These conditions imply that for $e\in E_0$,
  $-u(e)\leq f(\rev{e})\leq 0$.
  We extend the function $u$ to $E$ by setting $u(\rev{e})=0$ for all $e\in E_0$.
  Then, for all edges $e\in E$ we have $-u(\rev{e})\leq f(e)\leq u(e)$.
  The \emph{unit capacity} function satisfies $u(e)=1$ for all $e\in E_0$.

  The \emph{excess} $\exc_f(v)$ of a vertex $v\in V$ is defined as $\sum_{uv=e\in E} f(e)$. Due
  to anti-symmetry of $f$, $\exc_f(v)$ is equal to the amount of flow going into $v$ by the
  edges of $E_0$ minus the amount of flow going out of $v$ by the edges of $E_0$.
  The vertex $v\in V$ is called an \emph{excess} vertex if $\exc_f(v)>0$ and \emph{deficit} if
  $\exc_f(v)<0$.   Let $X$ be the set of excess vertices of $G$ and
  let $D$ be the set of deficit vertices.
  Define the \emph{total excess} $\totexc_f$ as the sum of excesses of the excess vertices, i.e.,
  $\totexc_f=\sum_{v\in X} \exc_f(v)=\sum_{v\in D}-\exc_f(v).$

  A flow $f$ is called a \emph{circulation} if there are no excess vertices,
  or equivalently, $\totexc_f=0$.

  Let $c:E_0\to \mathbb{Z}$ be the input \emph{cost} function.
  We extend $c$ to $E$ by setting $c(\rev{e})=-c(e)$ for all $e\in E_0$.
  The \emph{cost} $c(f)$ of a flow $f$ is defined as $\frac{1}{2}\sum_{e\in E} f(e)c(e)=\sum_{e\in E_0} f(e)c(e)$.

  \emph{To send a unit of flow} through $e\in E$ means to increase $f(e)$ by $1$ and simultaneously
  decrease $f(\rev{e})$ by $1$.
  By sending a unit of flow through $e$ we increase the cost of flow by $c(e)$.
  \emph{To send a unit of flow through a path $P$} means to send a unit of flow
  through each edge of $P$. In this case
  we also say that we \emph{augment flow $f$ along path $P$}.

  The \emph{residual network} $G_f$ of $f$ is defined as $(V,E_f)$, where $E_f=\{e\in E: f(e)<u(e)\}$.
  \paragraph{Price functions and distances.}
  We call any function $p:V\to \mathbb{R}$ a \emph{price function} on $G$.
  The \emph{reduced cost} of an edge $uv=e\in E$ wrt. $p$ is defined
  as $c_p(e):=c(e)-p(u)+p(v)$.
  We call $p$ a \emph{feasible price function} of $G$ if each edge $e\in E$
  has nonnegative reduced cost wrt. $p$.

  It is known that $G$ has no negative-cost cycles (negative cycles, in short) if and only if some feasible
  price function $p$ for $G$ exists.
  If $G$ has no negative cycles, distances in $G$ (where we interpret
  $c$ as a \emph{length} function) are well-defined.
  For $u,v\in V$, we denote by $\dist_G(u,v)$ the distance between $u$ and $v$,
  or, in other words, the length of a shortest $u\to v$ path in $G$.

  \begin{fact}\label{f:distanceto}
    Suppose $G$ has no negative cycles. Let $t\in V$ be reachable in $G$ from all vertices
    $v\in V$. Then the \emph{distance to} function
    $\dist_{G,t}(v):=\dist_{G}(v,t)$ is a feasible price function of~$G$.
  \end{fact}

  For $A,B\subseteq V(G)$ we sometimes write $\dist_G(A,B)$ to denote
  $\min\{\dist_G(u,v):u\in A, v\in B\}$.

\newcommand{\bnd}{\ensuremath{\partial}}
\newcommand{\pc}{\ensuremath{\mathcal{P}}}
  \paragraph{Planar graph toolbox.}
  An \emph{$r$-division} of a simple undirected plane graph $G$ is a collection of $O(n/r)$
  edge-induced subgraphs of $G$, called \emph{pieces}, whose union is $G$ and such
  that each piece $P$ has $O(r)$ vertices and $O(\sqrt{r})$ \emph{boundary vertices}.
  The boundary vertices $\bnd{P}$ of a piece $P$ are the vertices of $P$ shared with some other piece.

  An \emph{$r$-division with few holes} has an additional property
  that for each piece $P$, (1) $P$ is connected, (2) there exist $O(1)$ faces of $P$
  whose union of vertex sets contains $\bnd{P}$.

\newcommand{\dc}{\ensuremath{\text{DC}}}
\newcommand{\ddg}{\ensuremath{\text{DDG}}}

Let $G_1,\ldots,G_\lambda$ be some collection of plane graphs, where each $G_i$
has a distinguished boundary set $\bnd{G_i}$ lying on $O(1)$ faces of $G_i$.
A \emph{distance clique} $\dc(G_i)$ of $G_i$ is defined as a complete digraph on $\bnd{G_i}$
such that the cost of the edge $uv$ in $\dc(G_i)$ is equal to $\dist_{G_i}(u,v)$.
\begin{theorem}[MSSP~\cite{CabelloCE13, Klein05}]\label{t:mssp}
  Suppose a feasible price function on $G_i$ is given.
  Then the distance clique $\dc(G_i)$ can be computed in $O((|V(G_i)|+|E(G_i)|+|\bnd{G_i}|^2)\log{|V(G_i)|}))$ time.
\end{theorem}

The graph $\ddg=\dc(G_1)\cup \ldots \dc(G_\lambda)$ is called a \emph{dense distance graph}\footnote{Dense distance graphs
have been defined differently multiple times in the literature. We use the definition of \cite{GawrychowskiK18, nussbaum2014network} that captures
all the known use cases (see \cite{GawrychowskiK18} for discussion).}.
\begin{theorem}[FR-Dijkstra~\cite{FR06, GawrychowskiK18}]\label{t:fr}
  Given a feasible price function of $DDG$, single-source
  shortest paths in $\ddg$ can be computed in $O\left(\sum_{i=1}^\lambda |\bnd{G_i}|\frac{\log^2{n}}{\log^2\log{n}}\right)$ time, where $n=|V(DDG)|$.
\end{theorem}

  \paragraph{Scaling framework for minimum-cost circulation.}

  The following fact characterizes minimum circulations.
  \begin{fact}[\cite{circcycle}]\label{f:negcycle}
    Let $f$ be a circulation. Then $c(f)$ is minimum iff $G_f$ has no negative cycles.
  \end{fact}
  It follows that a circulation $f$ is minimum if there exists
  a feasible price function of $G_f$.

\begin{definition}[\cite{BertsekasT88, GoldbergT90, Tardos85}]
    A flow $f$ is $\eps$-optimal wrt. price function $p$ if
    for any $uv=e\in E_f$, $c(e)-p(u)+p(v)\geq -\eps.$
  \end{definition}
  The above notion of $\eps$-optimality allows us, in a sense,
  to measure the optimality of a circulation: the smaller $\eps$,
  the closer to the optimum a circulation $f$ is.
  Moreover, if we deal with integral costs, $\frac{1}{n+1}$-optimality
  is equivalent to optimality.

  \begin{restatable}[\cite{BertsekasT88, GoldbergT90}]{lemma}{lscale}\label{l:scale}
    Suppose the cost function has integral values. Let circulation
    $f$ be $\frac{1}{n+1}$-optimal wrt. some price function $p$.
    Then $f$ is a minimum cost circulation.
  \end{restatable}
  \begin{proof}
    Suppose $f$ is not minimum-cost.
    By Fact~\ref{f:negcycle}, $f$ is not minimum-cost iff $G_f$ contains
    a simple negative cycle $C$.
    Note that the cost of $C$ is the same wrt. to the cost functions
    $c$ and $c_p$, as the prices cancel out.
    Therefore $\sum_{e\in C}c_p(e)\geq -\frac{n}{n+1}>-1$.
    But the cost of this cycle is integral and hence is at least $0$,
    a contradiction.
  \end{proof}

    Let $C=\max_{e\in E_0}\{|c(e)|\}$.
  Suppose we have a procedure $\textsc{Refine}(G,f_0,p_0,\eps)$ that, given a circulation $f_0$ in $G$ that is $2\eps$-optimal wrt. $p_0$,
  computes a pair $(f',p')$ such that $f'$ is a circulation in $G$, and it is $\eps$-optimal wrt. $p'$.
  We use the general \emph{scaling} framework, due to Goldberg and Tarjan~\cite{GoldbergT90}, as given in Algorithm~\ref{alg:scaling}.
  By Lemma~\ref{l:scale}, it computes a min-cost circulation in $G$ in $O(\log(nC))$ iterations.
  Therefore, if we implement $\textsc{Refine}$ to run in $T(n,m)$ time,
  we can compute a minimum cost circulation in $G$ in $O(T(n,m)\log{(nC)})$ time.

\begin{algorithm}[]
  \begin{algorithmic}[1]
    \Procedure{MinimumCostCirculation}{$G$}
  \State $f(e):=0$ for all $e\in G$
  \State $p(v):=0$ for all $v\in V$
  \State $\eps:=C/2$
  \While{$\eps>\frac{1}{n+1}$}\Comment{$f$ is $2\eps$-optimal wrt. $p$}
    \State $(f,p):=\Call{Refine}{G,f,p,\eps}$
    \State $\eps:=\eps/2$
  \EndWhile
    \State \Return $f$ \Comment{$f$ is circulation $\frac{1}{n+1}$-optimal wrt. $p$, i.e., a minimum-cost circulation}
  \EndProcedure
\end{algorithmic}
  \caption{Scaling framework for min-cost circulation.\label{alg:scaling}}
\end{algorithm}
  \section{Refinement via Successive Approximate Shortest Paths}\label{s:general}
  In this section we introduce our implementation of $\textsc{Refine}(G,f_0,p_0,\eps)$.
  For simplicity, we start by setting $c(e):=c(e)-p_0(u)+p_0(v)$.
  After we are done, i.e., we have a circulation $f'$ that is $\eps$-optimal wrt. $p'$,
  (assuming costs reduced with $p_0$), we will return $(f',p'+p_0)$ instead.
  Therefore, we now have $c(e)\geq -2\eps$ for all $e\in E_{f_0}$.

  Let $f_1$ be the flow initially obtained from $f_0$ by sending a unit of
  flow through each edge $e\in E_{f_0}$ such that $c(e)<0$.
  Note that $f_1$ is $\eps$-optimal, but
  it need not be a circulation.

  We denote by $f$ the \emph{current flow} which we will gradually
  change into a circulation.
  Recall that $X$ is  the set of excess vertices of $G$ and
  $D$ is the set of deficit vertices (wrt. to the current flow $f$).
  Recall a well-known method of finding the min-cost circulation exactly~\cite{succsp1, succsp2, succsp3}:
  repeatedly send flow through shortest $X\to D$ paths in $G_f$.
  The sets $X$ and $D$ would only shrink in time.
  However, doing this on $G_f$ exactly would be too costly.
  Instead,
  we will gradually convert $f$ into a circulation,
  by sending flow from vertices of $X$ to vertices of $D$
  but only using approximately (in a sense) shortest paths.

  Let $\round(y,z)$ denote the smallest integer multiple of $z$ that is greater than $y$.

 For any $e\in E$, set
$c'(e)=\round(c(e)+\eps/2,\eps/2).$
  We define $G_f'$ to be the ``approximate'' graph $G_f$ with the costs given by $c'$ instead of $c$.

  For convenience, let us also define an extended version $G_f''$ of $G_f'$ to be
  $G_f'$ with two additional vertices $\source$ (a super-excess-vertex) and $\sink$ (a super-deficit-vertex) added.
  Let $M=\sum_{e\in E}|c'(e)|+\eps$.
  We also add to  $G_f''$ the following auxiliary edges:
  \begin{enumerate}
  \item an edge $v\sink$ for all $v\in V$,
      we set $c'(v\sink)=0$ if $v\in D$ and $c'(v\sink)=M$ otherwise,
  \item an edge $\source x$ with $c'(\source x)=0$ for all $x\in X$.
  \end{enumerate}
  Clearly, $\dist_{G_f''}(\source,\sink)=\dist_{G_f'}(X,D)$ and every vertex in $G_f''$ can reach $\sink$.

Our algorithm can be summarized very briefly, as follows. Start with $f=f_1$.
  While $X\neq \emptyset$, send a unit of flow along any shortest
  path $P$ from $X$ to $D$ in $G_f'$ (equivalently: from $\source$ to $\sink$ in $G_f''$).
  Once finished, return $f$ and $\dist_{G_f'',\sink}$ as the price function.
  The correctness of this approach follows from the following two facts
  that we discuss later on:
  \begin{enumerate}[label={(\arabic*)}]
    \item $G_f'$ is negative-cycle free at all times,
    \item after the algorithm finishes, $f$ is a circulation in $G$ that is $\eps$-optimal wrt. $\dist_{G_f'',\sink}$.
  \end{enumerate}

  If implemented naively, the algorithm would need $O(m)$ negative-weight shortest
  paths computations to finish.
  If we used Bellman-Ford method for computing shortest paths, the algorithm
  would run in $O(nm^2)$ time.
  To speed it up, we apply two optimizations.

  First, as in the successive shortest paths algorithm for general
  graphs~\cite{succsp4, succsp5}, we observe that the distances $\dist_{G_f'',\sink}$ computed
  before sending flow through a found shortest $\source\to\sink$ path
  constitute a feasible price function of $G_f''$ \emph{after} augmenting the flow.
  This allows us to replace Bellman-Ford algorithm with Dijkstra's algorithm
  and reduce the time to $O(m^2+nm\log{n})$.
  Next, instead of augmenting the flow along a single shortest $X\to D$ path,
  we send flow through a maximal set of edge-disjoint shortest $X\to D$ paths,
  as in Hopcroft-Karp algorithm for maximum bipartite matching~\cite{HopcroftK73}.
  Such a set can be easily found in $O(m)$ time when the distances to $\sink$ in $G_f''$ are known.
  This way, we finish after only $O(\sqrt{m})$ phases of shortest path computation and flow augmentation.
  The pseudocode is given in Algorithm~\ref{alg:refine}.

\begin{algorithm}[h]
  \begin{algorithmic}[1]
\Require{$f_0$ is a circulation in $G$ $2\eps$-feasible wrt. $p_0$}
    \Require{$\Call{DistancesTo}{H,t,p}$ computes the vector of distances (i.e., $\dist_{G,t}$) from all $v\in V(H)$ to $t\in V(H)$,
    where $p$ is a feasible price function of $H$.}
    \Require{$\Call{SendFlow}{f,E^*}$ returns a flow $f'$ such that $f'(e)$ equals $f(e)+1$ if $e\in E^*$,
    $f(e)-1$ if $\rev{e}\in E^*$, and $f(e)$ otherwise.}
    \Ensure{$(f,p)$, where $f$ is a circulation in $G$ $\eps$-feasible wrt. $p$}
    \Procedure{Refine}{$G,f_0,p_0,\eps$}
    \State $c(e):=c(e)-p_0(u)+p_0(v)$ for all $e=uv\in E$.
    \State $f:=\Call{SendFlow}{f_0,\{e\in E_{f_0}:c(e)<0\}}$
  \State $p(v):=0$ for all $v\in V$
    \While{$X\neq 0$}\Comment{$p$ is a feasible price function of $G_f''$}
    \State Construct $G_f''$ out of $G_f'$.
    \State $p:=\Call{DistancesTo}{G_f'',\sink,p}$\label{l:dijkstra}
    \State $Q_0,\ldots,Q_k:=$ a maximal set of edge-disjoint $\source\to\sink$ paths in $G_f''$ consisting solely \linebreak\phantom{very sorry for thishack}
    of edges satisfying $c'_p(e)=0$.\label{l:augment}
    \State $f:=\Call{SendFlow}{f,E((Q_0\cup\ldots\cup Q_k)\cap G_f')}$
  \EndWhile
    \State \Return $(f,\Call{DistancesTo}{G_f'',\sink,p}+p_0)$ \Comment{$f$ is $\eps$-feasible wrt. $\dist_{G_f'',\sink}+p_0$}
  \EndProcedure
\end{algorithmic}
  \caption{Refinement via successive shortest paths.\label{alg:refine}}
\end{algorithm}
\subsection{Analysis}
Below we state some key properties of our refinement method.
     \begin{lemma}\label{l:peq}
      Suppose $G_f''$ has no negative cycles. Then $f$ is $\eps$-optimal wrt. $\dist_{G_f'',\sink}$.
     \end{lemma}

  \begin{proof}
    Recall that $G_f$ and $G_f'$ have the same sets of edges, only different costs.
    Let $uv=e\in G_f$.
    Set $p:=\dist_{G_f'',\sink}$.
    By Fact~\ref{f:distanceto}, $c'(e)-p(u)+p(v)\geq 0$.
    Note that $c(e)\geq c'(e)-\eps$.
    Hence, $c(e)-p(u)+p(v)\geq c'(e)-p(u)+p(v)-\eps\geq -\eps.$
  \end{proof}
\begin{restatable}{lemma}{lpath}\label{l:path}
  If $X\neq \emptyset$,
  then there exists a path from $X$ to $D$ in $G_f$.
\end{restatable}

In order to prove Lemma~\ref{l:path}, we need the following Lemma of Goldberg et al.~\cite{GoldbergHKT17}.
\begin{lemma}[\cite{GoldbergHKT17}]\label{l:goldberg}
  Define $G^+_f=(V,E^+_f)$, where $E^+_f=\{e\in E: f(e)<f_0(e)\}.$
  Then for any $C\subseteq V$,
  $$\sum_{v\in C} \exc_f(v)\leq |\{ab=e\in E^+_f:a\in C, b\notin C\}|.$$
\end{lemma}

\begin{proof}[Proof of Lemma~\ref{l:path}]
  Observe that $G^+_f$ is a subgraph of $G_f$.
  It is hence enough to prove that there exists a $X\to D$ path in $G^+_f$.
  Let $Q$ be the set of vertices reachable from any vertex of $X$ in $G^+_f$.
  If $D\cap Q=\emptyset$, then $\sum_{v\in Q} \exc_f(v)=\sum_{v\in X}\exc_f(v)=\totexc_f\geq |X|>0$.
  By Lemma~\ref{l:goldberg}, there exists an edge $e=ab$ in $G^+_f$
  such that $a\in Q$ and $b\notin Q$.
  Hence $b$ is reachable from $X$ and $b\notin Q$, a contradiction.
\end{proof}

Before we proceed further, we need to introduce more notation.
Let $\Delta$ denote the length of the shortest $X\to D$
path in $G_f'$ ($\Delta$ changes in time along with $f$).

  Let $q=\totexc_{f_1}$. Clearly, $q\leq m$. For $i=1,\ldots,q$, denote by
  $f_{i+1}$ the flow (with total excess $q-i$) obtained from $f_i$
  by sending a unit of flow through an arbitrarily chosen shortest $X\to D$ path $P_i$ of $G_{f_i}'$.
  
  For $i=1,\ldots,q$, let $\Delta_i$ be the value $\Delta$ when $f=f_i$. We set $\Delta_{q+1}=\infty$.

    \begin{restatable}{lemma}{lcorrect}\label{l:correct}
      Let $p^*_i:V\cup\{\source,\sink\}\to \{k\cdot \eps/2:k\in \mathbb{Z}\}$ be defined as
    $p^*_i=\dist_{G_{f_i}'',\sink}$. Then:
    \begin{enumerate}[label={(\arabic*)}]
      \item $G_{f_{i}}'$ has no cycles of non-positive cost,
      \item for any $e\in P_i$, the reduced cost of $\rev{e}$ wrt. $p^*_i$ is positive,
      \item $p_i^*$ is a feasible price function of both $G_{f_i}''$ and $G_{f_{i+1}}''$,
      \item $0<\Delta_i\leq \Delta_{i+1}$.
    \end{enumerate}
    \end{restatable}
  \begin{proof}
    We proceed by induction on $i$. We also prove that (5) $G_{f_{i+1}}'$ has no
    $0$-cost cycles.

    Consider item (1).
    If $i=1$, then $G_{f_1}'$ has positive cost edges only so it cannot have
    non-positive cost cycles.
    Otherwise, if $i>1$, then by the inductive hypothesis $p^*_{i-1}$ is a feasible price
    function for $G_{f_i}''$, so $G_{f_i}''$ has no negative cycles.
    By item (5) for $i-1$, we also have that $G_{f_i}$ has no $0$-cost cycles.

    By (1) applied for $i$, $G_{f_i}''$ has no negative cycles,
    and thus distances in $G_{f_i}''$ are well-defined
    and so is $p^*_i$.
    Since the edge weights of $G_f''$ are integer multiples of $\eps/2$,
    all distances in this graph are like that as well.
    We conclude that the values of $p^*_i$ are indeed multiples of $\eps/2$.

    Note that $M$ is large enough so that
    a path $v\to \sink$ in $G_{f_i}''$ uses an in-edge of $\sink$ with weight $M$
    if and only if $v$ cannot reach $D$ in $G_{f_i}'$.
    Hence, if $v$ can reach $D$ in $G_{f_i}'$, 
    we have $p_i^*(v)=\dist_{G_{f_i}''}(v,\sink)=\dist_{G_{f_i}'}(v,D).$

    We now prove (2) and (3). By Fact~\ref{f:distanceto}, $p^*_i$ is a feasible price function of $G_{f_i}''$
    before sending flow through $P_i$.
    To prove that $p^*_i$ is a feasible price function of $G_{f_{i+1}}'$,
    we only need to consider the reduced costs of edges $\rev{e}$,
    where $uv=e\in P_i$.
    Since $e$ is on a shortest path from $X$ to $D$ in $G_{f_i}'$, by (2) and Fact~\ref{f:revbound} we have
    $$p_i^*(v)=\dist_{G_{f_i}'}(v,D_i)=\dist_{G_{f_i}'}(u,D_i)-c'(e)=p_i^*(u)-c'(e)<p^*_i(u)+c'(\rev{e})-\eps.$$
    Equivalently
    $c'(\rev{e})-p^*_i(v)+p^*_i(u)>\eps$
    so indeed $p^*_i$ remains feasible for $G_f'$ after sending flow through $P_i$.
    To see that $p^*_i$ is feasible for $G_{f_{i+1}}''$, note that in comparison
    to $G_{f_i}''$, $G_{f_{i+1}}''$
    has less auxiliary edges $sx$ and for one auxiliary edge $dt$ its cost
    is increased from $0$ to $M$.

    We clearly have $\Delta_1>0$ since $G_{f_1}'$ has positive edges only.
    Note that by Lemma~\ref{l:path}, $\Delta_i$ is finite, whereas $\Delta_{q+1}=\infty$.
    Hence, $\Delta_i\leq \Delta_{i+1}$ holds for $i=q$.
    Suppose $i<q$ and $\Delta_i>\Delta_{i+1}$.
    We have $p_i^*(\source)=\Delta_i$ and $p_i^*(\sink)=0$.
    By (3), $p_i^*$ is a feasible price function for $G_{f_{i+1}}''$, so
    $$0\leq \dist_{G_{f_{i+1}}''}(\source,\sink)-p^*_i(\source)+p^*_i(\sink)=\Delta_{i+1}-\Delta_i<0$$
    a contradiction. We have proved (4).

    Finally, we prove that (5) $G_{f_{i+1}}'$ has no $0$-cost cycles.
    Note that the cost of any cycle in $G_{f_i}'$ is preserved if we reduce the edge
    costs with $p^*_i$.
    Recall from (3) that the edge costs reduced by $p^*_i$ are all non-negative both in $G_{f_i}'$
    and $G_{f_{i+1}}'$.
    Hence all $0$-length cycles in $G_{f_{i+1}}'$
    consist solely of edges of reduced (with $p_i^*$) cost $0$.
    But $G_{f_{i+1}}'$ is obtained from $G_{f_i}'$ by replacing
    some edges with reduced cost $0$ with reverse
    edges with positive reduced cost.
    No such edge can thus lie on a $0$-cost cycle in $G_{f_{i+1}}'$.
    A $0$-cost cycle without an edge of $E(G_{f_{i+1}}')\setminus E(G_{f_i}')$
    cannot exist, as it would also exist in $G_{f_i}$ which would contradict (1).
  \end{proof}

By Lemmas~\ref{l:path}~and~\ref{l:correct}, our general algorithm
computes a circulation $f_{q+1}$ such that $p_q^*$ is a feasible
price function of $G_{f_{q+1}}'$.
Since $f_{q+1}$ has no negative cycles, by Lemma~\ref{l:peq}, $f_{q+1}$ is $\eps$-optimal wrt. $\dist_{G_f'',\sink}$.
We conclude that the algorithm is correct.

The following lemma is the key to the running time analysis.
\begin{restatable}{lemma}{lbound}\label{l:bound}
  If $X\neq\emptyset$ (equivalently, if $\Delta<\infty$), then
  $\totexc_f \cdot \Delta\leq 6\eps m.$
\end{restatable}
\begin{proof}
  This proof is inspired by the proof of an analogous fact in~\cite{GoldbergHKT17}.

  Suppose $f=f_i$ for some $i$, $1\leq i\leq q$ and let $p:=p_i^*$,
  where $p_i^*$ is as in Lemma~\ref{l:correct}.
  Moreover, by Lemma~\ref{l:correct}, $p(x)\geq \Delta$ for all $x\in X$
  and $p(d)\leq 0$ for all $d\in D$, where $\Delta>0$.

  Let $uv=e\in E^+$.
  Then since $e\in E_f$,
  and $f$ is $\eps$-optimal wrt. $p$ (by Lemma~\ref{l:peq}),
  $c(e)-p(u)+p(v)\geq -\eps$.
  Equivalently, $p(u)-p(v)\leq c(e)+\eps=-c(\rev{e})+\eps$.
  But since $\rev{e}\in E_{f_0}$, $-c(\rev{e})\leq 2\eps$ and
  we obtain $$p(u)-p(v)\leq 3\eps.$$

  For a number $z$, define $L(z)=\{v\in V: p(v)\geq z\}$.
  We have $X\subseteq L(z)$ for any $z\leq\Delta$ and $D\cap L(z)=\emptyset$ for any $z>0$.
  Consequently, for any $0< z\leq \Delta$
  we have $$\sum_{v\in L(z)} \exc_f(v)=\sum_{v\in X} \exc_f(v)=\totexc_f.$$
  Let $E^+(z)=\{uv=e\in E^+:p(u)\geq z, p(v)<z\}=\{uv=e\in E^+:u\in L(z), v\notin L(z)\}$.
  By Lemma~\ref{l:goldberg}, for any $0< z\leq \Delta$,
  $\totexc_f\leq |E^+(z)|.$

  For a particular edge $e\in E^+$, the condition $p(u)\geq z, p(v)<z$, by $z>p(v)\geq p(u)-3\eps$,
  is equivalent to $z\in (p(u)-3\eps,p(u)]$.
  Consider the sets $$E^+(\eps/2), E^+(\eps/2), E^+(2\cdot \eps/2),\ldots, E^+(k\cdot \eps/2), \ldots, E^+(\Delta).$$
  Since each edge $e\in E^+$ belongs to at most $6$ of these sets,
  the sum of their sizes is at most $12m$.
  Hence, for some $z^*$, where $0< z^*\leq \Delta$, $|E^+(z^*)|\leq \frac{6m\eps}{\Delta}$.
  We conclude $\totexc_f \Delta\leq 6m\eps$.
  \end{proof}

\subsection{Efficient Implementation}
As mentioned before, we could use Lemma~\ref{l:correct} directly:
start with flow $f_1$ and $p_0^*\equiv 0$.
Then, repeatedly compute a shortest $X\to D$ path $P_i$ along
with the values $p_i^*$ using Dijkstra's
algorithm on $G_f''$ (with the help of price function $p_{i-1}^*$ to make the
edge costs non-negative), and send flow through $P_i$ to obtain $f_{i+1}$.
However, we can also proceed as in Hopcroft-Karp algorithm
and augment along many shortest $X\to D$ paths of cost $\Delta$ at once.
We use the following lemma.
\begin{restatable}{lemma}{lincr}\label{l:incr}
  Let $p$ be a feasible price function of $G_f''$.
  Suppose there is no
  $\source\to\sink$ path in $G_f''$ consisting of edges with reduced (wrt. $p$) cost $0$.
  Then $\Delta=\dist_{G_f'}(X,D)>p(\source)-p(\sink).$
\end{restatable}
\begin{proof}
  By Lemma~\ref{l:path}, there exists some $X\to D$ path
  in $G_{f'}$ and hence also a $\source\to\sink$ path $e_1\ldots e_k$ in $G_f''$.
  Let $e_i=u_iu_{i+1}$, where $u_1=\source$ and $u_{k+1}=\sink$.
  We have $c'(e_i)-p(u_i)+p(u_{i+1})\geq 0$ for all $i$
  but for some $j$ we also have $c'(e_j)-p(u_j)+p(u_{j+1})>0$.
  So,
  $$\Delta=\dist_{G_f'}(X,D)=\dist_{G_f''}(\source,\sink)\geq \sum_{i=1}^k c'(e_i)=p(\source)-p(\sink)+\sum_{i=1}^k (c'(e_i)-p(u_i)+p(u_{i+1}))>p(\source)-p(\sink).\qedhere$$
\end{proof}

Suppose we run the simple-minded algorithm. Assume that at some point $f=f_i$, and we have $p_i^*$ computed.
Any $\source\to\sink$ path in $G_{f_i}''$ with reduced (wrt. $p_i^*$) cost $0$
corresponds to some shortest $X\to D$ path (of length $\Delta_i$)
in $G_f'$.
Additionally, we have $p_i^*(\source)=0$ and $p_i^*(\sink)=\Delta_i$.

Let $Q_0,\ldots,Q_k$ be some maximal set of edge-disjoint $\source\to\sink$ paths in $G_{f_i}''$ with reduced
cost $0$.
By Lemma~\ref{l:correct}, we could in principle choose $P_i=Q_0, P_{i+1}=Q_1, \ldots, P_{i+k}=Q_k$
and this would not violate the rule that we repeatedly choose
shortest $X\to D$ paths.

Moreover, $p_i^*$ is a feasible price function of $G_{f_{i+1}}''$ for
any choice of $P_i=Q_j$, $j=0,\ldots,k$.
Hence, the reduced cost wrt. $p_i^*$ of any $\rev{e}\in Q_j$, is non-negative.
Therefore, in fact $p_i^*$ is a feasible price
function of all $G_{f_{i+1}}'',G_{f_{i+2}}'',\ldots,G_{f_{i+k+1}}''$.
On the other hand, since for all $e\in P_i\cup\ldots\cup P_{i+k}$,
the reduced cost (wrt. $p_i^*$) of $\rev{e}$ is positive,
and the set $Q_0,\ldots,Q_k$ was maximal, we conclude
that there is no $\source\to\sink$ path in $G_{f_{i+k+1}}''$ consisting only
of edges with reduced cost (wrt. $p_i^*$) $0$.
But $p_i^*(s)-p^*(t)=\Delta_i$, so by Lemma~\ref{l:incr} we have
$\Delta_{i+k+1}>\Delta_i$.

Since we can choose a maximal set $Q_0,\ldots,Q_k$ using a DFS-style procedure in $O(m)$
time (for details, see Section~\ref{s:pathcycle}, where we take a closer look at it
to implement it faster in the planar case),
we can actually move from $f_i$ to $f_{i+k+1}$ and simultaneously
increase $\Delta$ in $O(m)$ time.
Since $p_i^*$ is a feasible price function of $G_{f_{i+k+1}}''$,
the new price function $p_{i+k+1}^*$ can be computed, again, using Dijkstra's algorithm.
The total running time of this algorithm is $O(m+n\log{n})$ times the number of
times $\Delta$ increases.
\begin{lemma}\label{l:change}
  The value $\Delta$ changes $O(\sqrt{m})$ times.
\end{lemma}
\begin{proof}
  By Lemma~\ref{l:correct}, $\Delta$ can only increase, and if it does, it increases by at least $\eps/2$.
  After it increases $2\sqrt{m}$ times, $\Delta\geq \eps\sqrt{m}$.
  But then, by Lemma~\ref{l:bound}, $\totexc_f$ is no more than $6\sqrt{m}$.
  As each change of $\Delta$ is accompanied with some decrease of $\totexc_f$,
  $\Delta$ can change $O(\sqrt{m})$ times~more.
\end{proof}

\begin{theorem}
  $\textsc{Refine}$ as implemented in Algorithm~\ref{alg:refine} runs in $O((m+n\log{n})\sqrt{m})$.
\end{theorem}

We can in fact improve the running time to $O(m\sqrt{m})$ by taking advantage of so-called Dial's
implementation of Dijkstra's algorithm~\cite{Dial69}.
The details can be found in Appendix~\ref{s:dial}.

\subsection{Bounding the Total Length of Augmenting Paths.}
\begin{fact}\label{f:revbound}
  For every $e\in E$ we have $c'(e)+c'(\rev{e})>\eps$.
  \end{fact}
  \begin{proof}
    We have $c'(e)>c(e)+\eps/2$.
    Hence, $c'(e)+c'(\rev{e})>c(e)+\eps/2+c(\rev{e})+\eps/2=\eps$.
  \end{proof}

There is a subtle reason why we set $c'(e)$ to be $\round(c(e)+\eps/2,\eps/2)$
instead of $\round(c(e),\eps)$.
Namely, this allows us to obtain the following bound.

\begin{lemma}\label{l:cmonoton}
  For any $i=1,\ldots,q$ we have $c'(f_{i+1})-c'(f_i)<\Delta_i-|P_i|\cdot \eps.$
\end{lemma}
\begin{proof}
  We have
  $$c'(f_{i+1})-c'(f_i)=\frac{1}{2}\sum_{e\in E}(f_{i+1}(e)-f_i(e))c'(e)=\frac{1}{2}\left(\sum_{e\in P_i}c'(e)-\sum_{e\in P_i}c'(\rev{e})\right)=\frac{1}{2}\sum_{e\in P_i}(c'(e)-c'(\rev{e})).$$
  By Fact~\ref{f:revbound}, $-c'(\rev{e})<c'(e)-\eps$ for all $e\in E$.
  Hence
  $$c'(f_{i+1})-c'(f_i)< \sum_{e\in P_i}c'(e)-|P_i|\cdot \eps=\Delta_i-|P_i|\cdot\eps.\qedhere$$
\end{proof}
\begin{lemma}\label{l:diffbound}
  Let $f^*$ be any flow. Then
  $c(f_0)-c(f^*)\leq 2\eps m.$
\end{lemma}
\begin{proof}
  We have
  $$c(f_0)-c(f^*)=\frac{1}{2}\sum_{e\in E} (f_0(e)-f^*(e))c(e).$$
  If $f_0(e)>f^*(e)$, then $\rev{e}\in E_{f_0}$ and hence $c(\rev{e})\geq -2\eps$, and
  thus $c(e)\leq 2\eps$.
  Otherwise, if $f_0(e)<f^*(e)$ then $e\in E_{f_0}$ and $c(e)\geq -2\eps$.

  In both cases $(f_0(e)-f^*(e))c(e)\leq 2\eps$.
  Therefore, since $|E|=2m$, $c(f_0)-c(f^*)\leq 2\eps m$.
\end{proof}
\begin{lemma}\label{l:apprbound}
  Let $f^*$ be any flow. Then
  $|c'(f^*)-c(f^*)|\leq \eps m.$
\end{lemma}
\begin{proof}
  Recall that we had $0< c'(e)-c(e)\leq \eps$. Hence $|f^*(e)(c'(e)-c(e))|\leq \eps$ and:
  $$|c'(f^*)-c(f^*)|=\frac{1}{2}\left|\sum_{e\in E}f^*(e)(c'(e)-c(e))\right|\leq \frac{1}{2}\sum_{e\in E} |f^*(e)(c'(e)-c(e))|\leq \frac{1}{2}\sum_{e\in E}\eps=\eps m.\qedhere$$
\end{proof}
The inequalities from Lemmas~\ref{l:cmonoton}, \ref{l:diffbound}~and~\ref{l:apprbound}
combined give us the following important property of the
set of paths we augment along.

\begin{lemma}\label{l:sumpath}
  The total number of edges on all the paths we send flow through is $O(m\log{m})$.
\end{lemma}
\begin{proof}
  By Lemma~\ref{l:diffbound} and the fact that $c(f_1)\leq c(f_0)$, we have:
  $$c(f_1)-c(f_{q+1})\leq {c(f_0)-c(f_{q+1})}\leq 2\eps m.$$
  On the other hand, by Lemma~\ref{l:apprbound} and Lemma~\ref{l:cmonoton}, we obtain:
  \begin{align*}
  c(f_1)-c(f_{q+1}) & \geq (c'(f_1)-\eps m)+(-c'(f_{q+1})-\eps m)
                    =c'(f_1)-c'(f_{q+1})-2\eps m\\
                    &=\sum_{i=1}^q (c'(f_i)-c'(f_{i+1}))-2\eps m
                    \geq\sum_{i=1}^q (|P_i|\cdot\eps-\Delta_i)-2\eps m.
  \end{align*}
  By combining the two inequalities and applying Lemma~\ref{l:bound}, we get:
  $$\sum_{i=1}^q |P_i|\leq 4m +\sum_{i=1}^q\frac{\Delta_i}{\eps}\leq 4m+\sum_{i=1}^q \frac{6m}{\totexc_{f_i}}=4m+\sum_{i=1}^q \frac{6m}{q-i+1}=O(m\log{m}).\qedhere$$
\end{proof}

\section{Unit-Capacity Min-Cost Circulation in Planar Graphs}\label{s:planar}
In this section we show that the refinement algorithm per scale from Section~\ref{s:general} can be simulated
on a planar digraph more efficiently.
Specifically, we prove the following theorem.
\begin{theorem}
  \textsc{Refine} can be implemented on a planar graph in $\Ot((nm)^{2/3})$ time.
\end{theorem}

\newcommand{\sg}{\ensuremath{\overline{G}}}
\newcommand{\pcg}{\ensuremath{\overline{\pc}}}

Let $r\in [1,n]$ be a parameter. Suppose we are given
an $r$-division with few holes $\pc_1,\ldots,\pc_\lambda$
of $G$ such that for any $i$ we have $\lambda=O(n/r)$, $|V(\pc_i)|=O(r)$, $|\bnd{\pc_i}|=O(\sqrt{r})$,
$\bnd{\pc_i}$ lies on $O(1)$ faces of $\pc_i$,
and the pieces are edge-disjoint.
We set $\bnd{G}=\bigcup_{i=1}^\lambda \bnd{\pc_i}$. Clearly, $|\bnd{G}|=O(n/\sqrt{r})$.

In Appendix~\ref{s:rdiv}, we show that we can reduce our instance to the case when the above
assumptions are satisfied in nearly-linear time.

Since $m$ might be $\omega(n)$, we cannot really guarantee that $|E(\pc_i)|=O(r)$,
This will not be a problem though, since, as we will see, for all the computations
involving the edges of $\pc_i$ (e.g., computing shortest
paths in $\pc_i$, or sending a unit of flow through a path of $\pc_i$) of all edges $uv=e\in E(\pc_i)$ we will only care about an edge
$e\in E(\pc_i)\cap G_f$ with minimal cost $c'(e)$.
Therefore, since $\pc_i$ is planar, at any time only $O(r)$ edges of $\pc_i$ will be needed.

Recall that the per-scale algorithm for general graphs (Algorithm~\ref{alg:refine})
performed $O(\sqrt{m})$ phases, each consisting
of two steps: a shortest path computation (to compute the price function $p^*$ from Lemma~\ref{l:correct}),
followed by the computation of a maximal set of edge-disjoint
augmenting paths of reduced (wrt. $p^*$) cost $0$.
We will show how to implement both steps in $\Ot(n/\sqrt{r})$ amortized time,
at the additional total data structure maintenance cost (over all phases) of $\Ot(mr)$.
Since there are $O(\sqrt{m})$ steps, this will yield
$\Ot(nm)^{2/3})$
time by appropriately setting $r$.

We can maintain the flow $f$ explicitly, since it undergoes
only $O(m\log{n})$ edge updates (by Lemma~\ref{l:sumpath}).
However, we will not compute the entire price function $p^*$
at all times explicitly, as this is too costly.
Instead, we will only compute $p^*$ limited to the subset $\bnd{G}\cup \{\source,\sink\}$.

For each $\pc_i$, define $\pc_{f,i}'=G_f'\cap \pc_i$.
We also define $\pc_{f,i}''$ to be $\pc_{f,i}'$ with vertices $\{\source,\sink\}$ added,
and those edges $\source v$, $v\sink$ of $G_{f}''$
that satisfy $v\in V(\pc_i)\setminus \bnd{\pc_i}$.
This way, $\pc_{f,i}''\subseteq G_f''$ and $E(\pc_{f,i}'')\cap E(\pc_{f,j}'')=\emptyset$
for $i\neq j$.
The costs of edges $e\in E(\pc_{f,i}'')$
are the same as in $G_{f}''$, i.e., $c'(e)$.
Besides, for each $i$ we will store a ``local'' price function $p_i$ that
is feasible only for~$\pc_{f,i}''$,

After the algorithm finishes, we will know how the circulation looks like precisely.
However, the general scaling algorithm requires us to also output price
function $p$ such that $f$ is an $\eps$-optimal circulation wrt. $p$.
$f$ is $\eps$-optimal wrt. $p^*$ in the end, but we will
only have it computed for the vertices $\bnd{G}\cup\{\source,\sink\}$.
Therefore, we extend it to all remaining vertices of~$G$.

\begin{restatable}{lemma}{local}\label{l:local}
  Suppose we are given the values of $p^*$ on $\bnd{\pc_i}$ and
  a price function $p_i$ feasible for $\pc_{f,i}''$.
  Then we can compute
  the values $p^*(u)$ for all $v\in V(\pc_{f,i}'')$ in $O(r\log{r})$ time.
\end{restatable}
\begin{proof}
  Recall that $p^*(v)=\dist_{G_f''}(v,\sink)$ for all $v\in V$.
  Since we already know the values of $p^*$ on $\bnd{\pc_i}\cup \{\source,\sink\}$,
  we only need to compute them for the vertices $V_i=V(\pc_{f,i}'')\setminus(\bnd{\pc_i}\cup \{\source,\sink\})$.
  Take some shortest $v\to \sink$ path $\pc_v$ in $G_f''$, where $v\in V_i$.
  $\pc_v$ either contains some first vertex $b\in \bnd{\pc_i}$, or
  is fully contained in $\pc_{f,i}''$.
  In the former case we have $p^*(v)=\dist_{G_f''}(v,\sink)=p^*(b)+\dist_{\pc_{f,i}''}(v,b)$,
  and in the latter $p^*(v)=\dist_{\pc_{f,i}''}(v,\sink)$.
  Hence, to compute distances from all $v\in V_i$ to $\sink$
  in $G_f''$, it is sufficient to compute shortest paths
  on the graph $\pc_{f,i}^*$, defined as $\pc_{f,i}''$ with edges $b\to\sink$ of cost $c'(b\sink)=p^*(b)$,
  for all $b\in \bnd{\pc_i}$, added.

  To do that efficiently, note that $p_i$ is an almost feasible price function
  of $\pc_{f,i}^*$: the only edges with possibly negative reduced cost wrt. $p_i$ are
  the incoming edges of $\sink$.
  Therefore, we can still compute the distances to $\sink$ in $\pc_{f,i}^*$ in $O(r\log{r})$ time
  using a variant of Dijkstra's algorithm \cite{DinitzI17} that allows $k$
  such ``negative'' vertices if we want to compute single-source shortest paths in $O(km\log{n})$ time.
\end{proof}

Hence, in order to extend $p^*$ to all vertices of $G$
once the final circulation is found, we apply Lemma~\ref{l:local}
to all pieces.
This takes $O\left(\frac{n}{r}\cdot r\log{r}\right)=O(n\log{n})$ time.
\subsection{Dijkstra Step}\label{s:planar-dijkstra}
Let us start with an implementation of the Dijkstra step computing the new price function~$p^*$.
First, for each piece $\pc_i$ we define the compressed version $H_{f,i}''$ of $\pc_{f,i}''$ as follows.
Let
$V(H_{f,i}'')=\bnd{\pc_i}\cup \{\source,\sink\}$.
The set of edges of $H_{f,i}''$ is formed by:
\begin{itemize}
  \item a distance clique $\dc(\pc_{f,i}'')$ between vertices $\bnd{\pc_i}$ in $\pc_{f,i}''$,
  \item for each $v\in \bnd{\pc_i}$, an edge $\source v$ of cost $\dist_{\pc_{f,i}''}(\source, v)$ if this distance is finite,
  \item for each $v\in \bnd{\pc_i}$, an edge $v\sink $ of cost $\dist_{\pc_{f,i}''}(v, \sink)$ if this distance is finite,
  \item an edge $\source\sink$ of cost $\dist_{\pc_{f,i}''}(\source, \sink)$ if this distance is finite.
\end{itemize}

Recall that we store a price function $p_i$ of $\pc_{f,i}''$.
Therefore, by Theorem~\ref{t:mssp}, $\dc(\pc_{f,i}'')$ can be computed
in $O(r\log{r})$ time.
All needed distances $\dist_{\pc_{f,i}''}(\source, v)$ and $\dist_{\pc_{f,i}''}(v, \sink)$ can be computed
in $O(r\log{r})$ time using Dijkstra's algorithm (again, with the help of price function $p_i$).

Now define $H_f''$ to be $\bigcup_{i=1}^\lambda H_{f,i}''$ with edges $\source v$ and $v\sink$
of $G_f''$ that satisfy $v\in \bnd{G}$ added.
\begin{fact}\label{f:translate}
  For any $u,v\in V(H_f'')$, $\dist_{H_f''}(u,v)=\dist_{G_f''}(u,v)$.
\end{fact}

Observe that $H_f''$ is a dense distance graph in terms of the definition of Section~\ref{s:prelims}:
it consists of $O(n/r)$ distance cliques $\dc(\pc_{f,i}'')$ with $O(\sqrt{r})$ vertices each,
and $O(n/\sqrt{r})$ additional edges which also can be interpreted as
$2$-vertex distance cliques.

Hence, given a feasible price function on $H_f''$,
we can compute distances to $\sink$ in $H_f''$ on it using Theorem~\ref{t:fr} in
$O\left(n/\sqrt{r}\frac{\log^2{n}}{\log^2\log{n}}\right)$ time.
Since $V(H_f'')=\bnd{G}\cup \{\source,\sink\}$, the price function $p^*$ we have is indeed sufficient.
The computed distances to $\sink$ form the new price function $p^*$
on $\bnd{G}\cup\{\source,\sink\}$ as in the algorithm for general graphs (see Algorithm~\ref{alg:refine}).

\subsection{Sending Flow Through a Path}\label{s:sendflow}
In the general case updating the flow after an augmenting path has been found was trivial.
However, as we operate on a compressed graph, the update procedure
has to be more involved.

Generally speaking, we will repeatedly find some shortest $\source\to\sink$ path $Q=e_1\ldots e_k$ in $H_f''$,
translate it to a shortest $\source\to\sink$ path $\pc$ in $G_f''$ and send flow
through it.
It is easy to see by the definition of $H_f''$ that $Q$ can be translated to a
shortest $\source\to \sink$ path in $G_f''$ and vice versa.
Each edge $e_j$ can be translated to either some subpath inside a single graph $\pc_{f,i}''$,
or an edge of $G_f''$ of the form $\source v$ or $v\sink$, where $v\in \bnd{G}$.
This can be
done in $O(r\log{n})$ time by running Dijkstra's algorithm
on $\pc_{f,i}''$ with price function~$p_i$.
We will guarantee that path $P$ obtained by concatenating
the translations of individual edges $e_j$ contains no repeated
edges of $G_f''$.

We now show how to update each $H_{f,i}''$ after sending flow through the found path $P$.
Note that we only need to update $H_{f,i}''$ if
$E(P)\cap E(\pc_{f,i}'')\neq \emptyset$.
In such case we call $\pc_i$ an \emph{affected piece}.
Observe that some piece can be affected
at most $O(m\log{m})$ times
since the total number of edges on all shortest augmenting paths $P$ in the entire algorithm,
regardless
of their choice, is $O(m\log{m})$ (see Lemma~\ref{l:sumpath}).

To rebuild $H_{f,i}''$ to take into account the flow augmentation we will need a feasible price function on $\pc_{f,i}''$
after the augmentation.
However, we cannot be sure that what we have, i.e., $p_i$,
will remain a good price function of $\pc_{f,i}''$ after the augmentation.
By Lemma~\ref{l:correct}, luckily, we know that
$p^*$
is a feasible price function after the augmentation
for the whole graph $G_{f}''$.
In particular, $p^*$ (before the augmentation) limited to $V(\pc_{f,i}'')$ is a feasible
price function of $\pc_{f,i}''$ after the augmentation.
Hence, we can compute new $p_i$ equal to $p^*$ using Lemma~\ref{l:local}
in $O(r\log{r})$ time.
Given a feasible price function $p_i$ on $\pc_{f,i}''$ after $f$ is augmented, we can recompute $H_{f,i}''$
in $O(r\log{r})$ time as discussed in Section~\ref{s:planar-dijkstra}.
We conclude that the total time needed to update the graph $H_f''$ subject to flow
augmentations is $O(mr\log{r}\log{m})=O(mr\log{n}\log{m})$.
\subsection{A Path Removal Algorithm}\label{s:pathcycle}
In this section we consider an abstract ``path removal'' problem,
that generalizes the problem of finding a maximal set of edge-disjoint
$\source\to\sink$ paths.
We will use it to reduce the problem of finding such a set of paths
on a subgraph of $G_f''$ consisting of edges with reduced cost $0$ wrt. $p^*$
to the problem of finding such a set of paths on
the zero-reduced cost subgraph of $H_f''$.

Suppose we have some directed acyclic graph $H$ with a fixed source $\source$
and sink $\sink$, that additionally undergoes some limited adversarial changes.
We are asked to efficiently support a number of \emph{rounds},
until $\sink$ ceases to be reachable from $\source$. Each round goes as follows.
\begin{enumerate}[label={(\arabic*)}]
\item We first find either any $\source\to\sink$ path $P$,
    or detect that no $\source\to\sink$ path exists.
\item Let $E^+\subseteq V\times V$, and $P\subseteq E^-\subseteq E(H)$ be some adversarial sets of edges. Let
    $H'=(V,E')$, where $E'=E(H)\setminus E^-\cup E^+$.
    Assume that for any $v\in V(H)$, if $v$ cannot reach $\sink$ in $H$,
    then $v$ cannot reach $\sink$ in $H'$ either.
    Then the adversarial change is to remove $E^-$ from $E$ and add $E^+$ to $E$, i.e., set $E(H)=E'$.
\end{enumerate}

Let $\bar{n}=|V(H)|$ and let $\bar{m}$ be the number of edges ever seen by the algorithm,
i.e., the sum of $|E(H)|$ and all $|E^+|$.
We will show an algorithm that finds all the paths $P$ in $O(\bar{n}+\bar{m})$ total time.
Let us also denote by $\bar{\ell}$ the sum of lengths of all returned paths $P$.
Clearly, $\bar{\ell}\leq \bar{m}$.

A procedure handling the phase (1) of each round, i.e., finding a $\source\to\sink$
path or detecting that there is none, is given in Algorithm~\ref{alg:maxpath}.
The second phase of each round simply modifies the representation of the graph $H$
accordingly.
Throughout all rounds, we store a set $W$ of vertices $w$ of $H$ for which
we have detected that there is no more $w\to \sink$ path in~$H$.
Initially, $W=\emptyset$.
Each edge $e\in E(H)$ can be \emph{scanned} or \emph{unscanned}.
Once $e$ is scanned, it remains scanned forever.
The adversarial edges $E^+$ that are inserted to $E(H)$ are initially unscanned.


\begin{algorithm}[h]
\small
\begin{algorithmic}[1]
  \Procedure{FindPath}{$H$}
  \State $Q:=$ an empty path with a single endpoint $\source$\Comment{$Q$ is an $\source\to \source$ path}

  \While{$\source\notin W$ and the other endpoint $y$ of $Q$ is not equal to $\sink$}\Comment{$Q$ is an $\source\to y$ path}
    \If{there exists an unscanned edge $yv=e\in E(H)$ such that $v\notin W$}\label{l:search}
      \State mark $e$ scanned
      \State $Q:=Qe$\label{l:append}
    \Else
      \State $W:=W\cup \{y\}$\label{l:winsert}
      \State remove the last edge of $Q$ unless $Q$ is empty\label{l:remove}
    \EndIf
  \EndWhile
  \If{$Q=\emptyset$}
    \State report $\sink$ not reachable from $\source$ and stop
  \Else
    \State \Return $Q$ and $Q:=0$.
  \EndIf
\EndProcedure
\end{algorithmic}
  \caption{Path-finding procedure. Returns a $\source\to\sink$ path in $H$ or detects that there is none.\label{alg:maxpath}}
  \end{algorithm}
The following lemmas establish the correctness and efficiency
of the crucial parts of Algorithm~\ref{alg:maxpath}.

\begin{restatable}{lemma}{blockingflow}
  Algorithm~\ref{alg:maxpath} correctly finds an $\source\to\sink$ path in $H$ or detects there is none.
\end{restatable}
\begin{proof}
First note that no edge is appended to $Q$ twice throughout all rounds:
only unscanned edges are ever appended to $Q$ and
are marked scanned immediately afterwards.
Hence, the algorithm stops.

Since $H$ is acyclic, $Q$ remains simple at all times.
Moreover, for each scanned edge $uv=e\in E(H)$ we either have $e\in Q$
  or $v\in W$.

The next observation is that immediately after line~\ref{l:winsert} is executed,
  for all edges $yv\in E(H)$ we have $v\in W$.
By the previous observation,
for all edges $e=yv\in E(H)$, we have either $v\in W$
or $e\in Q$.
But after line~\ref{l:winsert} is executed, $y$ is the other endpoint of $Q$,
so if $e\in Q$, then $y$ also appears somewhere earlier in $Q$, i.e.,
$Q$ is not simple, a contradiction.

Next we prove that $W$ contains only vertices $v$
  that cannot reach $\sink$.
Consider the first moment when some vertex $v\in W$ can
actually reach $\sink$ in $H$.
If this is a result of changing the edge set, this means
that $v$ cannot reach $\sink$ in $H$, but can reach $\sink$
  in $(V,E')$. This, however, violates our assumption about $(V,E')$.
So $v$ is the first vertex that gets inserted
to $W$ in line~\ref{l:winsert}, but actually can reach $\sink$ in $H$ at this time.
In this case, for all edges $vw\in E(H)$, $w\in W$ and $w$ was inserted into $W$ before $v$.
Therefore, $v$ has only edges to vertices that cannot reach $t$, and thus it
cannot reach $\sink$ itself, a contradiction.

  Let us also note that for each edge $uv\in E(H)$, $u\in W$ implies
  that $v$ cannot reach $t$.
  Otherwise, $u$ could in fact reach $u$, which would contradict our
  previous claim.

Next we show that if a run of the procedure does not find a $\source\to\sink$ path,
it visits only vertices $v$ (i.e., $v\in V(Q)$
at some point of that run) reachable from $\source$, and out of
  those visits all that can reach $\sink$.
Clearly the procedure does not visit any $v$ not reachable from $\source$,
  as in that case we would have $v\in V(Q)$ at some point,
  but $Q$ is always a path starting at $\source$, i.e.,
  all vertices of $Q$ are reachable from $\source$.
  Now suppose the procedure does not visit some $v$ that
is reachable from $\source$ and can reach $\sink$, and choose $v$
to be such that $\dist_{H}(\source,v)$ is minimum.
Clearly, $v\neq \source$.
Let $w$ be such vertex that $\dist_{H}(\source,v)=\dist_{H}(\source,w)+1$ and $e=wv\in E(H)$.
  Observe that $e$ is unscanned, as otherwise we would either have $e\in Q$ (and thus $v$ would
  be visited) or $v\in W$ (and thus $v$ would not reach $\sink$).
  Note that $w$ is never inserted into $W$, since that would imply that $v$ cannot reach $\sink$.
  Since $w$ is reachable from $\source$, can reach $\sink$ (because it can reach $v$),
  and $\dist_{H}(\source,w)<\dist_{H}(\source,v)$,
  $w$ is visited by the procedure.
  But since the procedure does not terminate prematurely having found a $\source\to\sink$
  path $P$, the edge $e$, being unscanned, will be appended to $Q$ in step (a) when $Q$ is a $\source\to w$
  path. Hence, $v$ will be visited, a contradiction.

Finally, the procedure either finds a $\source\to\sink$ path, or
proves that $\sink$ is not reachable from $\source$.
\end{proof}
\begin{restatable}{lemma}{stepremove}\label{lem:stepremove}
  The total number of times line~\ref{l:remove} is executed, through all rounds, is $O(\bar{n})$.
\end{restatable}
\begin{proof}
  Each execution of line~\ref{l:remove} is preceded by an insertion
  of some vertex to $W$.
  Each $v\in V(H)$ is inserted into $W$ at most once:
only the other endpoint of $Q$ can be inserted into $W$, and no
vertex of $W$ is ever appended to $Q$.
\end{proof}

\begin{restatable}{lemma}{stepappend}\label{lem:stepappend}
  Line~\ref{l:append} of Algorithm~\ref{alg:maxpath} is executed $O(\bar{n}+\bar{\ell})$ times through all rounds.
\end{restatable}
\begin{proof}
  By Lemma~\ref{lem:stepremove}, $O(\bar{n})$ edges $e$ appended to $Q$
  that are later popped in line~\ref{l:remove}.
  If the appended edge is never popped in step~\ref{l:remove}, it is a part
  of a returned path or cycle -- this happens precisely $\bar{\ell}$ times.
\end{proof}

\begin{restatable}{lemma}{stepunscanned}\label{l:stepunscanned}
  The total time used by Algorithm~\ref{alg:maxpath}, through all rounds, is $O(\bar{n}+\bar{m})$.
\end{restatable}
\begin{proof}
  We represent $W$ as a bit array of size $\bar{n}$.
  Then, by Lemmas~\ref{lem:stepremove} and~\ref{lem:stepappend},
  to show that the algorithm runs in $O(\bar{n}+\bar{m})$ time, we only
  need to implement line~\ref{l:search} so that its all
  executions take $O(\bar{m})$ time in total.
  But this is easy:
  it is sufficient to store the outgoing edges of each vertex $v$
  in a linked list, so that adding/removing edges takes $O(1)$ time
  and we can move to a next unscanned edge in $O(1)$ time.
\end{proof}

\subsection{Finding a Maximal Set of Shortest Augmenting Paths}\label{s:planarpaths}

Recall that for a general graph, after computing the price function $p^*$
we found a maximal set of edge-disjoint $\source\to\sink$ paths
in the graph $Z_f''$, defined as a subgraph of $G_f''$ consisting
of edges with reduced cost $0$ (wrt. $p^*$).
To accomplish that, we could in fact use the path removal algorithm from Section~\ref{s:pathcycle} run on $Z_f''$:
until there was an $\source\to\sink$ path in $Z_f''$, we would find
such a path $P$,
remove edges of $P$ (i.e., set $E^-=P$ and $E^+=\emptyset$),
and repeat.
Since in this case we never add edges, the assumption that $\sink$
cannot become reachable from any $v$ due to updating $Z_f''$ is met.

Let $Y_f''$ be the subgraph of the graph $H_f''$ from Section~\ref{s:planar-dijkstra}
consisting of edges with reduced (wrt. $p^*$) cost $0$.
Since all edges of $H_f''$ correspond to shortest paths in $G_f''$,
all edges of $Y_f''$ correspond to paths in $G_f''$ with reduced cost $0$.
Because $Z_f''$ is acyclic by Lemma~\ref{l:correct}, $Y_f''$ is acyclic as well.
Moreover, for any two edges $e_1,e_2\in E(Y_f'')$, if there is a path going
through~both $e_1$ and $e_2$ in $Y_f''$, then the paths represented
by $e_1$ and $e_2$ are edge-disjoint in $Z_f''$ (as otherwise $Z_f''$
would have a cycle).
Therefore, any path $Q$ in $Y_f''$ translates to a \emph{simple} path
in $Z_f''\subseteq G_f''$.

We will now explain why running Algorithm~\ref{alg:maxpath} on $Y_f''$ can be
used to find a maximal set of edge-disjoint $\source\to\sink$ paths.
Indeed, by Fact~\ref{f:translate}, $Y_f''$ contains an $\source\to\sink$
path iff $Z_f''$ does.
Since $Y_f''$ is just a compressed version of $Z_f''$, and $Z_f''$ undergoes
edge deletions only (since we only remove the found paths),
the updates to $Y_f''$ cannot make some $\sink$ reachable from some new vertex $v\in V(Y_f'')$.
Technically speaking, we should think of $Y_f''$ as undergoing both edge insertions
and deletions: whenever some path $Q\subseteq Y_f''$ is found,
we include $Q$ in $E^-$ and send the flow through a path corresponding to $Q$ in $G_f''$,
as described in Section~\ref{s:sendflow}.
But then for all affected pieces $\pc_i$, $H_{f,i}''$ is recomputed and thus
some of the edges of $Q$ might be reinserted to $Y_f''$ again.
These edges should be seen as forming the set $E^+$,
whereas the old edges of the recomputed graphs $H_{f,i}''$ belong to $E^-$.
In terms of the notation of Section~\ref{s:pathcycle}, when running
Algorithm~\ref{alg:maxpath} on $Y_f''$, we have
$\bar{n}=O(n/\sqrt{r})$.
The sum of values $\bar{\ell}$ from Section~\ref{s:pathcycle} over all phases of the algorithm is,
by Lemma~\ref{l:sumpath}, $O(m\log{m})$.
Similarly, again by Lemma~\ref{l:sumpath}, the sum of the values $\bar{m}$ from Section~\ref{s:pathcycle} over all phases, is $O(m^{3/2}+mr^2\log{m})$
(since each time $E^+$ might be as large as $r^2$ times the number
of affected pieces).

Recall that there are $O(\sqrt{m})$ phases, and the total time needed to maintain the
graph $H_f''$ subject to flow augmentations is $O(mr\log{r}\log{m})$ (see Section~\ref{s:sendflow}).
For each phase,
running a Dijkstra step to compute $p^*$ using FR-Dijkstra, followed by running
Algorithm~\ref{alg:maxpath} directly until there are
no $\source\to\sink$ paths in $Y_f''$ would lead to
$O\left(\sqrt{m}\left(\frac{n}{\sqrt{r}}\frac{\log^2{n}}{\log^2\log{n}}\right)+m^{3/2}+mr^2\log{m}\right)$ total time,
i.e., would not yield any improvement over the general algorithm.
However, we can do better by implementing Algorithm~\ref{alg:maxpath}
on $Y_f''$ more efficiently.



\newcommand{\prowmap}{\ensuremath{\pi}}

\newcommand{\mmat}{\ensuremath{\mathcal{M}}}

The following lemma is essentially proved in~\cite{ItalianoKLS17, Karczmarz18}.
However, as we use different notation, we give a complete proof below.

\begin{restatable}[\cite{ItalianoKLS17, Karczmarz18}]{lemma}{lreach}\label{l:reach}
  Let $Z$ be the subgraph of $\pc_{f,i}'$ consisting of edges
  with reduced cost $0$ wrt. to some feasible price function $p$.
  There exists $O(\sqrt{r})$ pairs
  of subsets 
  $(A_{i,1},B_{i,1}),(A_{i,2},B_{i,2}),\ldots$ of $\bnd{\pc_i}$ such that
  for each $v\in \bnd{\pc_i}$:
  \begin{itemize}
    \item The number of sets $A_{i,j}$ ($B_{i,j}$) such that $v\in A_{i,j}$ ($v\in B_{i,j}$, resp.) is $O(\log{r})$.
    \item Each $B_{i,j}$ is totally ordered according to some order $\prec_{i,j}$.
    \item For any $j$ such that $v\in A_{i,j}$, there exist
      $l_{i,v,j},r_{i,v,j}\in B_{i,j}$ such that
      the subset $R_{i,v}$ of $\bnd{\pc_i}$ reachable from $v$ in $Z$ can be expressed
      as
      $\bigcup_{j:v\in A_{i,j}} \{w\in B_{i,j}: l_{i,v,j}\preceq_{i,j} w\preceq_{i,j} r_{i,v,j}\}.$
  \end{itemize}
  The sets $A_{i,j}$, $B_{i,j}$ and the vertices $l_{i,v,j},r_{i,v,j}$ for all
  $v,j$ can be computed in $O(\sqrt{r}\log{r})$ time based
  on the distance clique between $\bnd{\pc_i}$ in $\pc_{f,i}'$ and the values of $p^*$ on $\bnd{\pc_i}$.
\end{restatable}
\begin{proof}
  The distance clique of $\bnd{\pc_i}$ in $\pc_{f,i}'$ can be partitioned
  into $O(\sqrt{r})$ rectangular Monge matrices
  $\mmat_1,\ldots,\mmat_q$, where $R_j$ and $C_j$ denote the sets
  of rows and columns, respectively, of $\mmat_i$, such that:
  \begin{itemize}
    \item these matrices have $O(\sqrt{r}\log{r})$ rows and columns in total,
    \item each $v\in \bnd{\pc_i}$ is a row of $O(\log{r})$ matrices $\mmat_j$,
    \item for all elements $\mmat_j[u,v]$, where $u\in R_j$, $v\in C_j$ we have $\mmat_j[u,v]\geq \dist_{\pc_{f,i}'}(u,v)$,
    \item for all $u,v\in \bnd{\pc_i}$ there exists such $\mmat_j$
      that $\mmat_j[u,v]=\dist_{\pc_{f,i}'}(u,v)$.
  \end{itemize}
  Recall that the Monge property here says that for any two rows $a,b\in R_j$
  and any two columns $x,y\in C_j$, such that $a$ is to the left of $b$,
  and $x$ is above $y$, we have $\mmat_j[a,x]+\mmat_j[b,y]\geq \mmat_j[a,y]+\mmat_j[b,x]$.

  The partition can be computed in $O(r\log{r})$ time when constructing the
  distance clique.
  The proof of the above can be found in \cite{GawrychowskiK18, MozesW10}.


  Denote by $\mmat_{j,p}$ the matrix with entries $\mmat_{j,p}[u,v]=\mmat_j[u,v]-p(u)+p(v)$.
  $\mmat_{j,p}$ is also Monge (see e.g., \cite{GawrychowskiK18}).
  Clearly, the non-infinite entries of each matrix $\mmat_{j,p}$ are non-negative,
  since $\dist_{\pc_{f,i}'}(u,v)-p(u)+p(v)\geq 0$ for all $u,v\in \bnd{\pc_i}$.
  For each $\mmat_{j,p}$ we find:
  \begin{itemize}
    \item the subset $B_{i,j}\subseteq C_j$ of its columns $b$ such that $\mmat_j[a,b]-p(a)+p(b)=0$ for some $a\in R_j$,
    \item the subset $A_{i,j}\subseteq R_j$ of its rows $a$ such that $\mmat_j[a,b]-p(a)+p(b)=0$ for some $b\in C_j$.
  \end{itemize}
  Both $A_{i,j}$ and $B_{i,j}$ can be found
  by finding row/column minima of $\mmat_{j,p}$ using SMAWK algorithm \cite{AggarwalKMSW87} in $O(|R_j|+|C_j|)$ time.
  Next, we again use SMAWK algorithm to find for each $a\in A_{i,j}$ the leftmost row minimum
  $\mmat_{j,p}[a,l_{i,a,j}]$
  and the rightmost row minimum $\mmat_{j,p}[a,r_{i,a,j}]$ of the row $a$ of $\mmat_{j,p}$.
  This takes $O(|A_{i,j}|+|C_j|)$ time as well.
  Set $\prec_{i,j}$ to be the order of columns in $\mmat_j$ restricted to $B_{i,j}$.

  For brevity, below set $A_j:=A_{i,j}$, $B_j:=B_{i,j}$, $l_{u,j}:=l_{i,u,j}$, $r_{u,j}:=r_{i,u,j}$ and $\prec_{j}:=\prec_{i,j}$.

  It is sufficient to show that for $u,v\in \bnd{\pc_i}$, a path $u\to v$ exists
  in $Z$ if and only if for some $j$, $u\in A_j$, $v\in B_j$ and
  $l_{u,j}\preceq_{j} v\preceq_{j} r_{u,j}$.
  Let us start with $\implies$ direction.
  There exists such $\mmat_j$ that $\dist_{\pc_{f,i}'}(u,v)=\mmat_j[u,v]$.
  Since $\dist_{\pc_{f,i}'}(u,v)-p(u)+p(v)=0$, we have $\mmat_{j,p}[u,v]=0$.
  Since $\mmat_{j,p}$ has non-negative entries, $u\in A_{j}$ and $v\in B_{j}$.
  By the definition of $l_{u,j}$ and $r_{u,j}$, $l_{u,j}\preceq_{j} v\preceq_{j} r_{u,j}$.

  Now suppose $u\in A_{j}$, $v\in B_{j}$ and $l_{u,j}\preceq_j v\preceq_j r_{u,j}$.
  Clearly $\mmat_{j,p}[u,v]\geq 0$. Suppose $\mmat_{j,p}[u,v]>0$. Then $l_{u,j}\prec_j v \prec_j r_{u,j}$.
  Since $v\in B_j$, there exists some row $x\neq u$ of $\mmat_{j,p}$ such that
  $\mmat_{j,p}[x,v]=0$.
  If the row $x$ is above $u$, by Monge property we have $0=\mmat_{j,p}[x,v]+\mmat_{j,p}[u,r_{u,j}]\geq \mmat_{j,p}[x,r_{u,j}]+\mmat_{j,p}[u,v]>0$,
  a contradiction.
  Similarly, if $x$ is below $u$,
  then $0=\mmat_{j,p}[u,l_{u,j}]+\mmat_{j,p}[x,v]\geq \mmat_{j,p}[u,v]+\mmat_{j,p}[x,l_{u,j}]>0$, a contradiction.
  So in fact $\mmat_{j,p}[u,v]=0$ and therefore a path $u\to v$ exists in $Z$.

  Clearly, the total size of sets $A_j,B_j$ is $O(\sqrt{r}\log{r})$ and
  the total time to find these subsets and all $l_{a,j},r_{a,j}$,
  is $O(\sqrt{r}\log{r})$, given the preprocessed matrices $\mmat_1,\ldots,\mmat_q$.
\end{proof}


%
%

Recall that in Section~\ref{s:pathcycle}, to bound the total running
time, it was enough to bound the total time spent on executing
lines~\ref{l:search}, \ref{l:append}~and~\ref{l:remove}.
We will show that using Lemma~\ref{l:reach}, in terms of the notation from Section~\ref{s:pathcycle}, we can
make the total time spent on executing line~\ref{l:search} only $\Ot(\bar{n}+\bar{\ell})$
instead of $O(\bar{m})$, at the cost of increasing the total time
of executing line~\ref{l:remove} to $\Ot(\bar{n})$.

Specifically, at the beginning of each phase we compute the data from Lemma~\ref{l:reach}
for all pieces $\pc_i$.
Since for all $i$ we have the distance cliques $\dc(\pc_{f,i}'')$ computed,
this takes $O\left(\frac{n}{r}\cdot \sqrt{r}\log{r}\right)=O(n/\sqrt{r}\log{n})$ time.
We will also recompute the information of Lemma~\ref{l:reach} for an affected piece $\pc_i$ after
$H_{f,i}''$ is recomputed.
As the total number of times some piece is affected is $O(m\log{m})$, this takes $O(m\sqrt{r}\log{r}\log{m})$ time
through all phases.

Whenever the data of Lemma~\ref{l:reach} is computed for some piece $\pc_i$,
for each pair $(A_{i,j},B_{i,j})$ we store $B_{i,j}\cap W$ in a dynamic predecessor/successor
data structure $D_{i,j}$, sorted by $\prec_{i,j}$.
For each $v\in\bnd{\pc_i}$ and $j$ such that $v\in A_{i,j}$ we store a vertex
$next_{i,v,j}$ initially equal to $l_{i,v,j}$.
It is easy to see that these auxiliary data structures can be constructed
in time linear in their size, i.e., $O(\sqrt{r}\log{r})$
time.
Hence, the total cost of computing them is
$O(\sqrt{m}n/\sqrt{r}\log{n}+m\sqrt{r}\log{r}\log{m})=O\left(\sqrt{m}\left(\frac{n}{\sqrt{r}}\frac{\log^2{n}}{\log^2\log{n}}\right)+mr\log{n}\log{m}\right)$.

Now, to implement line~\ref{l:remove}, when $y$ is inserted into $W$ we go through
all pieces $\pc_i$ such that $y\in \bnd{\pc_i}$ and all
$B_{i,j}$ such that $y\in B_{i,j}$.
For each such $(i,j)$, we remove $y$ from $D_{i,j}$ in $O(\log\log{n})$ time.
Recall that the sum of numbers of such pairs $(i,j)$ over all $v\in \bnd{G}$
is $O(\sum_{i=1}^\lambda |\bnd{\pc_i}|\log{r})=O(n/\sqrt{r}\log{n})$.
Hence, by Lemma~\ref{lem:stepremove} the total time spent on executing line~\ref{l:remove}
in a single phase is $O(n/\sqrt{r}\log{n}\log\log{n})$.

Finally, we implement line~\ref{l:search} as follows.
The unscanned edges of $Y_f''$ that are not between boundary vertices are handled
in a simple-minded way as in Lemma~\ref{l:stepunscanned}.
There are only $O(n/\sqrt{r})$ of those, so we can neglect them.
In order to be able to efficiently find some unscanned edge $yv$ such that $y,v\in \bnd{G}$
and $v\notin W$, we keep for any $v\in \bnd{G}$ a set $U_v$ of pieces
$\pc_i$ such that $v\in\bnd{\pc_i}$ and there may still be some
unscanned edges from $v$ to $w\in \bnd{\pc_i}$ in $H_{f,i}''$.
Similarly, for each $\pc_i\in U_v$ we maintain
a set $U_{v,i}$ of data structures $D_{i,j}$ such that $next_{i,v,j}\neq\nil$.
Whenever the data of Lemma~\ref{l:reach} is computed for $\pc_i$, $\pc_i$ is inserted back to $U_v$
for all $v\in \bnd{\pc_i}$, and the sets $U_{v,i}$ are recomputed with
no additional asymptotic overhead.
To find an unscanned edge $yv$, for each $\pc_i\in U_y$ we proceed as follows.
We attempt to find an unscanned edge $yv$ in $\pc_i$.
If we succeed or $U_y$ is empty, we stop. Otherwise
we remove $\pc_i$ from $U_y$ and repeat, i.e., try another $\pc_j\in U_y$, unless $U_y$
is empty.
To find an unscanned edge $yv$ from a piece $\pc_i$,
we similarly try to find an unscanned edge $yv$ in subsequent
data structures $D_{i,j}\in U_{v,i}$,
and remove the data structures for which we fail from $U_{v,i}$.
For a single data structure $D_{i,j}$,
we maintain an invariant that an edge $yw$, $w\in D_{i,j}$ has been scanned
iff $w\prec_{i,j} next_{i,v,j}$.
Hence, to find the next unscanned edge, we first find
$x\in D_{i,j}$ such that $next_{i,v,j}\preceq_{i,j} x$ and $x$ is smallest possible.
This can be done in $O(\log\log{n})$ time since $D_{i,j}$ is
a dynamic successor data structure.
If $x$ does not exist or $r_{i,v,j}\prec x$, then,
by Lemma~\ref{l:reach}, there are no more unscanned edges
$yw$ such that $w\in D_{i,j}$, and thus we remove $D_{i,j}$ from $U_{v,i}$.
Otherwise, we return an edge $yx$ and set $next_{i,v,j}$ to
be the successor of $x$ in $D_{i,j}$ (or possibly $next_{i,v,j}:=\nil$
if none exists), again in $O(\log\log{n})$ time.

Observe that all ``failed'' attempts to find an edge $yv$, where $v\in \bnd{G}$
can be charged to an insertion of some $\pc_i$ to $U_y$
or to an insertion of some $D_{i,j}$ to $U_{y,i}$.
The total number of such insertions is again $O\left(\sqrt{m}\frac{n}{\sqrt{r}}\log{n}+m\sqrt{r}\log{r}\log{m}\right)$.
A successful attempt, on the other hand, costs $O(\log\log{n})$ worst-case time.
Since line~\ref{l:search} is executed $O(\sqrt{m}n/\sqrt{r}+m\log{n})$ times
through all phases,
the total time spent on executing line~\ref{l:search} is again
$O\left(\sqrt{m}\left(\frac{n}{\sqrt{r}}\frac{\log^2{n}}{\log^2\log{n}}\right)+mr\log{n}\log{m}\right)$. By setting
$r=\frac{n^{2/3}}{m^{1/3}}\cdot\left(\frac{\log{n}}{\log{m}\cdot \log^2\log{n}}\right)^{2/3}$ we obtain the main result of this paper.

\begin{theorem}
The min-cost circulation in a planar multigraph can be found in
  \linebreak $O\left((nm)^{2/3}\cdot\frac{\log^{5/3}{n}\log^{1/3}{m}}{\log^{4/3}\log{n}}\cdot \log{(nC)}\right)$ time.
\end{theorem}

\bibliography{../references2}

\newpage

\appendix

\section{Reducing to the Case with an $r$-division with Few Holes}\label{s:rdiv}

\begin{theorem}[\cite{KleinMS13}]\label{t:rdiv}
  Let $G$ be a simple triangulated connected plane graph with $n$ vertices.
  For any $r\in [1,n]$, an $r$-division with few holes of $G$
  can be computed in $O(n)$ time.
\end{theorem}

Let $\sg$ be an undirected simple plane graph obtained from $G_0$ by subsequently,
(1) ignoring the directions of edges, (2) removing multiple edges (i.e., leaving
at most one, arbitrary edge $uv$ for any $\{u,v\}\subseteq V$), and (3) embedding $\sg$ into plane,
(4) triangulating the faces of $\sg$ using infinite-cost edges.
We will never send flow through ``dummy'' infinite-cost edges; we use them to guarantee some
useful topological properties of the pieces.

We build an $r$-division $\pcg_1,\ldots,\pcg_\lambda$ with few holes of $\sg$
using Theorem~\ref{t:rdiv}.
For each $i$ we have $|V(\pcg_i)|=O(r)$, $|E(\pcg_i)|=O(r)$ and $|\bnd{\pcg_i}|=O(\sqrt{r})$.
At this point each $uv=e\in E(\pcg)$ can be contained in many pieces.
We choose one piece $\pcg_{\{u,v\}}$ containing $e$ and make the cost of
$e$ infinite in all the others, effectively turning $e$ into a dummy edge
in those pieces.

Now we go back to our original graph $G$.
We obtain pieces $\pc_1,\ldots,\pc_\lambda\subseteq G$ as follows.
For each $uv=e\in E$, we make $e$ the edge of $\pc_i$ such that $\pcg_{\{u,v\}}=\pcg_i$
and make $e$ inherit the embedding of $uv\in E(\sg)$.
Similarly, for each of the dummy edges $uv=e'\in E(\sg)$, we direct it arbitrarily
and make it an edge of $\pc_i$ such that $e'\in E(\pcg_i)$.
We set $\bnd{\pc_i}=\bnd{\pcg_i}$.
This way, the properties of an $r$-division with few holes: (1) $|V(\pc_i)|=O(r)$, (2) $|\bnd{\pc_i}|=O(\sqrt{r})$,
and (3) $\bnd{\pc_i}$ lies on $O(1)$ faces of $\pc_i$,
are still satisfied.
The only difference is that now $|E(\pc_i)|$ might be $\omega(r)$.
In Section~\ref{s:planar} we have already justified that this is not a big problem, though.

\newcommand{\tin}{\text{in}}
\newcommand{\tout}{\text{out}}
\newcommand{\told}{\text{old}}
\newcommand{\indeg}{\text{indeg}}
\newcommand{\outdeg}{\text{outdeg}}
\newcommand{\mindeg}{\text{mindeg}}

\section{\textsc{Refine} in $O(m^{3/2})$ time}\label{s:dial}

So-called Dial's implementation of Dikjstra's algorithm~\cite{Dial69} can compute
the distances to a single sink $\sink$ to all vertices satisfying $\dist_{G}(v,\sink)\leq K\eps$
in $O(m+K)$ time, assuming all costs are non-negative integer multiples of $\eps$.
Hence, if we are given a possibly negatively-weighted graph with a price function $p$,
in $O(m+K)$ we can compute the distances to $\sink$ from all vertices such that
$\dist_{G}(v,\sink)-p(v)+p(\sink)\leq K\eps$.
Unfortunately, we cannot use Dial's algorithm directly, since the (reduced) distances to $\sink$
in $G_f''$ can be generally $\omega(m)$.

However, one can observe two things. First, by Lemma~\ref{l:bound}, $\Delta\leq 6\eps m$.
Moreover, in the implementation we do not need to use the price function $p^*$ from Lemma~\ref{l:correct},
which we do in line~\ref{l:augment}.
In fact, by Lemma~\ref{l:incr}, any feasible price function $p$ of $G_f''$ will do, provided
that $p(s)-p(t)=\Delta$.

We will maintain the invariant that $p(\sink)\leq 0$, and $p(\source)=0$.
Clearly, the invariant is satisfied initially.
In line~\ref{l:dijkstra}, we run Dijkstra's algorithm
with price function $p$
and stop it when it visits $\source$.
Since $p(\source)=0$, and Dijkstra's algorithm visits vertices $v$ in non-decreasing
order of values $\dist_{G_f''}(v,\sink)-p(v)+p(\sink)$, for any visited $v$ we have
$$\Delta=\dist_{G_f''}(\source,\sink)\geq \dist_{G_f''}(\source,\sink)+p(\sink)=\dist_{G_f''}(\source,\sink)-p(\source)+p(\sink)\geq \dist_{G_f''}(v,\sink)-p(v)+p(\sink).$$
Hence, indeed such a Dijkstra run can be performed in $O(m+\Delta/\eps)=O(m)$ time
using Dial's implementation.
Next, we set $p(v):=\dist_{G_f''}(v,\sink)-\dist_{G_f''}(\source,\sink)$ for
all visited $v$, whereas for the unvisited vertices $v$ we leave $p(v)$ unchanged.
Afterwards, we have $p(\source)=0$ and $p(\sink)=-\Delta<0$.

We need to verify that $p$ remains a feasible price function after it is altered.
Let $U$ be the set of visited vertices.
First note that before the substitution, for any $u\in U$ we have
$\dist_{G_f''}(u,\sink)-\dist_{G_f''}(\source,\sink)\leq p(u)-p(\source)=p(u)$.
Therefore, for each $u\in U$, its price cannot increase.
Now, let $uv=e\in E(G_f'')$.
If $\{u,v\}\subseteq U$, then the reduced cost of $e$ is non-negative, since
$p$ is a shifted distance-to function $\dist_{G_f'',t}$ on these vertices.
If $v\notin U$, then $c'_p(e)$ cannot decrease due to substitution,
and we had $c'_p(e)\geq 0$ before, so $c'_p(e)\geq 0$ afterwards as well.
Finally, suppose $u\notin U$ and $v\in U$.
Then, since $u$ was not visited before Dijkstra's run was terminated,
we had $(c'(e)-p(u)+p(v))+\dist_{G_f''}(v,\sink)-p(v)+p(\sink)\geq \dist_{G_f''}(\source,\sink)-p(\source)+p(\sink)$,
or equivalently $c'(e)-p(u)+(\dist_{G_f''}(v,\sink)-\dist_{G_f''}(\source,\sink))\geq -p(\source)=0$.
But $p(u)$ is not changed afterwards, and $p(v):=\dist_{G_f''}(v,\sink)-\dist_{G_f''}(\source,\sink)$,
so $c'_p(e)\geq 0$ afterwards as well.
So indeed $p$ remains feasible.

\end{document}